\documentclass[11pt,letter]{article}

\usepackage{amsmath}
\usepackage{xspace}
\usepackage{amssymb}
\usepackage{epsfig}
\usepackage{graphicx}

\newtheorem{Thm}{Theorem}
\newtheorem{Lem}[Thm]{Lemma}

\newtheorem{Obs}{Observation}

\newtheorem{Fact}{Fact}
\newtheorem{Case}{Case}
\newenvironment{proof}{\noindent {\textbf{Proof }}}{$\Box$ \medskip}

\begin{document}
\title{Fixed Parameter Tractable Algorithm for Firefighting Problem}
\author{Ming Lam Leung\thanks{Department of Computer Science and Engineering, The Chinese University of Hong Kong, Shatin, Hong Kong SAR, China}}
\maketitle


\abstract{The firefighter problem is defined as below. A fire initially breaks out at a vertex r on a graph G. In each step, a firefighter chooses to protect one vertex, which is not yet burnt. And the fire spreads out to its unprotected neighboring vertices afterwards. The objective of the problem is to choose a sequence of vertices to protect, in order to save maximum number of vertices from the fire.\\

In this paper, we will introduce a parameter k into the firefighter problem and give several FPT algorithms using a random separation technique of Cai, Chan and Chan. We will prove firefighter problem is FPT on general graph if we take total number of vertices burnt to be a parameter. If we parameterize the number of protected vertices, we discover several FPT algorithms of the firefighter problem on degree bounded graph and unicyclic graph. Furthermore, we also study the firefighter problem on weighted and valued graph, and the problem with multiple fire sources on degree-bounded graph.}

\section{Introduction}
The {\bf firefighter problem} is a discrete-time game on a graph G defined as below. Initially, a fire breaks out at a vertex r of G. For each time step, a firefighter choose one vertex not yet on fire to protect it (that vertex remains protected thereafter), and then the fire spreads from the vertices on fire to all of their unprotected neighbours. This process continues round by round. The game ends when the fire cannot spread anymore, i.e. all of the neighbours of the burnt vertices are already protected. Then the vertices not on fire are considered saved. The objective of the problem is to choose a sequence of vertices to protect, in order to save maximum number of vertices in the graph from the fire.\\

The firefighter problem is introduced by Hartnell \cite{Har95} in 1995 and can be used to model the spread of fire, diseases, computer viruses in a marco-control level. The firefighter problem is shown to be NP-complete even for trees of maximum degree 3 by Finbow et al. \cite{FKMR07} Actually, Hartnell and Li \cite{HL9} proved that a simple greedy method for trees is a 0.5-approximation algorithm, and MacGillivray and Wang \cite{MW03} have solved the problem in polynomial time in some subclass of trees. Recently, Cai, Verbin and Yang \cite{CVY08} have obtained a (1-1/$e$)-approximation algorithm for trees based on a linear programming relaxation and randomized rounding, and they have also considered fixed parameter tractability of the problem on trees. Besides, Cai and Wang \cite{C09} also study the defending ability of a graph as a whole by considering the average percentage of vertices the firefighter can save. Furthermore, various aspects of the problem for d-dimensional grids have been considered by Develin and Hartke \cite{DH07}, Fogarty \cite{Fog03}, Wang and Moeller \cite{WM02}, and Cai and Wang \cite{C09}, among others.\\

In this report, we consider fixed parameter tractability of the firefighter problem on various graphs other than trees or  d-dimensional-grid. We establish several FPT algorithms of the problem, either use number of vertices burnt or number of vertices protected to be parameter. Our main results are:

\begin{enumerate}
\item Parameter k on number of vertices burnt
\begin{enumerate}
\item Firefighter problem on general graph is FPT
\end{enumerate}
\item Parameter k on number of vertices protected (i.e. k = no. of rounds)
\begin{enumerate}
\item Firefighter problem on degree-bounded graph is FPT

\item  Firefighter problem on unicyclic graph is FPT

\end{enumerate}
\item Reduction of Firefighter problem on weighted and valued graph
\item Firefighter problem with multiple fire sources, multiple protection and multiple burning can be reduced to the original firefighter problem
\end{enumerate}

The main tool we use is the $random\;separation$ method of Cai, Chan and Chan \cite{CCC06}.  In the rest of the paper, we fix our notations and give definitions in Section 2. We then show that the parameterized firefighter problem with number of vertices burnt as parameter is FPT. We then establish FPT algorithms on degree bounded graph (Section 4) and unicyclic graph (Section 5). We also study the possibility of generalizing the algorithm on unicyclic graph to the family of tree+b edges. In Section 6, we will discuss the possibility of finding a solution in the general firefighter problem on the weighted and unequal valued graph, using local replacement. Finally, we will explore further in the firefighter problem with multiple fire sources, multiple protection and multiple burning in Section 7.


\section{Definitions and Notation}
Here we define some terms. Let $G$ be a undirected graph with a source vertex $s$, which is the origin of the fire. A vertex is $burnt$ once it is on fire, and $protected$ once it is protected by the firefighter, and $saved$ if it is not burnt at the end of the game. We assume the game ends at time step $t$, which means that the fire can no longer spread to any unburnt vertices after the firefighter protected t vertices. A $strategy$ for the $Firefighter$ problem is a sequence $\{v_1,v_2,...,v_t\}$ of protected vertices of $G$ such that vertex $v_i$ is protected at time $i$, where $1 \le i \le t$. An $optimal\;strategy$ of a $Firefighter$ problem is a $strategy$ maximizing the total number of vertices saved on the graph $G$.\\

We let $V$ and $E$ be the vertex set and edge set of graph $G$, and we let $n$ and $m$ be the number of vertices and edges of the graph $G$ respectively, i.e. $n=|V|, m=|E|$. For any set of vertices $V'$, We denote $N(V')$ to be the set of neighbours of $V'$. We also define $N[V']=N(V') \cup V'$.\\

Define $E(V')$ be the set of edges in the induced subgraph containing vertices $V'$, and $E(V_1,V_2)$ be the set of cross edges linking the vertex set $V_1$ and $V_2$. Define $E[V_1,V_2] = E(V_1,V_2) \cup E(V_1) \cup E(V_2)$.\\

We define $merge(G;V')$ to be a graph $G'$ constructed by merging the vertices in set $V'$ in $G$ into one single vertex. For example, $merge(G;$\{u,v\}$)$ is a graph constructed by merging the vertices $u$ and $v$ into one vertex.\\

A graph $G$ is said to be $degree\;bounded$ if the degree of each vertex on $G$ is bounded by a fixed constant $d$. 
A graph $G$ is a $unicyclic$ graph if the graph contains only 1 cycle, i.e. it has n vertices and n edges. Note that a $unicyclic$ graph can also be seen as a tree plus one extra edge.\\

In this paper, we mainly consider FPT algorithms for the following three versions of the $Firefighter$ problem. 
\begin{enumerate}
\item \textbf{At Most $k$ Vertices Burnt}: At most $k$ vertices are burnt at the end of the game. We ask if there is a strategy saving at least $n-k$ vertices.

\item \textbf{Exactly $k$ Vertices Burnt}: The parameter $k$ is the number of vertices burnt at the end of the game. We ask if there is a strategy saving at least $n-k$ vertices.

\item \textbf{Maximum $k$-Vertex Protection}: The parameter $k$ is the upper bound of number of protected vertices. We need to find a strategy to protect at most $k$ vertices in order to maximize the number of vertices saved.

The main tool we use is the $random\;separation$ method of Cai, Chan and Chan \cite{CCC06}. This technique produces randomized algorithms in FPT time. We can derandomize the algorithms using $universal\;sets$ \cite{NSS95}. Naor et al. give a construction of a (n,t)-universal set of cardinality $2^tt^{O(\log t)}\log n$ in time $2^tt^{O(\log t)}n \log n$. A ($n$,$t$)-$universal$ set is a set of binary vectors of length $n$ while all the $2^t$ configurations appear in the set for every subset of size t of the indices.\\

\end{enumerate}
We are going to construct a FPT algorithm to solve the problems \textbf{At Most $k$ Vertices Burnt} and \textbf{Exactly $k$ Vertices Burnt} on general graph in Section 3. And then we will focus on the problem \textbf{Maximum $k$-Vertex Protection} in Section 4 and 5, which is proven to be W[1]-hard in general graph. We will introduce FPT algorithm to solve it on degree bounded graph in Section 4, and show that it is FPT on unicyclic graph in Section 5. Besides, we study the reduction of the problem on weighted and valued graph and discuss about the case with multiple fire sources, multiple protection and multiple burning in Section 6 and 7.\\

The complexity of all FPT algorithms introduced in this report are summarized in Table 1.\\

\begin{table}[ht]
\caption{Summary of FPT algorithms} 
\centering 
\begin{tabular}{c c c c} 
\hline\hline 
Problem & Randomised & Deterministic  &  \\ [0.5ex] 
\hline 
\textbf{At Most $k$ Vertices Burnt} & $O(4^kn)$  & $2^{O(k)}n \log n$ \\ 
\textbf{Exactly $k$ Vertices Burnt} & $O(4^kn)$  & $2^{O(k)}n \log n$ \\
\textbf{Maximum $k$-Vertex Protection} \\
Degree-Bounded Graph & O($3^{k^2(d+1)}(n+m)$) & --- \\
Unicyclic Graph & $k^{-O(k)}n^3$ & $k^{-O(k)}n^3\log n$ \\
Tree+b edges& $k^{-O(k)}n^{(2b+1)}$ & $k^{-O(k)}n^{(2b+1)}\log n$ \\
\hline 
\end{tabular}
\label{table:nonlin} 
\end{table}


\section{At Most k Vertices Burnt on general graph}
In this section, we are considering the \textbf{At Most k Vertices Burnt} problem on general graph. We call a strategy S satisfying if it saves at least $n-k$ vertices, i.e. the total number of vertices burnt is at most $k$. The goal of the problem is to find a satisfying strategy if one exists.\\

The algorithm first colors each vertex of G randomly and independently by either green or red with probability $1/2$. We call a coloring of G a $good\;coloring$ if there exists a satisfying strategy $S'$ such that all vertices in $S'$ are green, and all the burnt vertices are red.\\

\begin{center}
\includegraphics{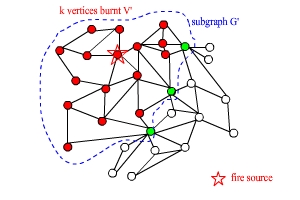}\\
\underline{Figure 1: A good coloring example of \textbf{At Most k Vertices Burnt}}
\end{center}

Given a good coloring of $G$, we can find the specific satisfying strategy $S'$ as the following. First, starting from the source s, do breath-first-search only on the red vertices to get a red subtree $V'$ rooted at s. It takes linear time on number of vertices and edges of the subgraph $V'$. After the BFS, all of the neighbours of $V'$ have to be green, and $V^*=N(V')$ must be a satisfying strategy if the coloring is good. If $|V'|\le k$, we check if the green vertex set $V^*=N(V')$ is a valid strategy set for the problem using the procedure below:

\begin{center}
\fbox{
\begin{minipage}[l1pt]{5.50in}
{\bf  Procedure for verifying if $V^*$ is a satisfying strategy set} \\

{\bf Input}: graph G, source s, vertex set $V^*$ to be verified, and burnt vertex set $V'$\\
{\bf Output}: If true, output the sequence of satisfying strategy in correct order. If false, output "no".
\begin{enumerate}
	\item Do BFS on the subgraph $V'$ starting from s, find the distance $d[s,v]$ of the each vertices $v$ on $V^*$ from the source $s$ on the subgraph $G'=\{N[V'],E(V')\cup E(V',V^*)\}$.  
	\item Sort the vertices $V^*$ according to $d[s,v]$ in ascending order. Store the sorted sequence in an array $\{v_i\}$ 
	\item for i = 1 to $|V^*|$ do\\
                      	 If the distance from the source of the i-th vertex $d[s,v_i]< i$, then the solution is not valid, output "no" and halt.
            \item Output the sequence $\{v_i\}$ as the corresponding firefighting strategy in correct order.
    
\end{enumerate}
\end{minipage}
}
\end{center}

Notice that $N[V']=V'\cup N(V')= V'\cup V^*$ in the context here. Therefore, the subgraph $G'=\{N[V'],E(V')\cup E(V',V^*)\}$ in step 1, is nothing but the induced subgraph $V'\cup V^*$ with the edge set $E(V^*)$ removed.

The above procedure is correct based on the following fact.

\begin{Fact} If $S^*=\{v_i\}$ is a satisfying strategy in correct order, the distance between the source $s$ and the i-th protected vertices $v_i$ on the subgraph $G'=\{V'\cup S^*,E(V')\cup E(V',S^*)\}$ has to be at least $i$.
\end{Fact}

Fact 1 is obviously true as the fire will spread out by one unit in the subgraph $G'$, after the firefighter protect one vertex at each time step. If the distance between $s$ and the i-th protected vertices $v_i$ on the subgraph $G'$ is less than $i$, the vertex $v_i$ will be burnt by the fire before time step $i$.\\

Before calculating the complexity of the above algorithm, we have the following observation.

\begin{Obs}
If the number of the vertices burnt is upper bounded by k, there exists a satisfying strategy $S'$ with at most k protected vertices. i.e. $|S'|\le k$.
\end{Obs}

The above statement is true because the number of protected vertices is equal to the total time step t. Moreover, each time step at least one vertex is burnt if the game is not ended. Therefore, $|S'|=t \le k$. \\

Since the size of $V'$ and $S'$ are both bounded by k, we have a good coloring with $V'$ colored red and $S'$ colored green with probability at least $2^{-2k}$. Since the BFS algorithm and distance finding is of linear time, and the sorting procedure in step 2 of the verifying procedure takes only O($k\log k$) time, therefore the \textbf{Procedure for verifying if $V^*$ is a satisfying strategy set} takes linear time. After all, there exists a randomized FPT algorithm solving the problem in O($2^{2k}(n+m)$) time.\\ 

Moreover, we can use $\{n,2k\}$ universal set to derandomize it and get a deterministic FPT algorithm in O($2^{2k}(n+m) \log (n+m)$) time. This result matches the previous result of the model \textbf{Saving All But $k$ Vertices} in the special case of firefighter on tree \cite{CVY08}, which gives us a result of $2^{O(k)}n \log n$ time as $m=n-1$ for trees. \\

Furthermore, we can modify the above algorithm a little bit to solve the \textbf{Exactly $k$ Vertices Burnt} problem with the same time complexity. In the above algorithm, we run the \textbf{Procedure for verifying if $V^*$ is a satisfying strategy set} if $|V'|\le k$. If we modify the algorithm as we run the procedure only if $|V'|= k$, then we can solve the \textbf{Exactly $k$ Vertices Burnt} problem on general graph with the same time complexity.


\section{Maximum $k$-Vertex Protection on degree bounded graphs}

In this section, we are considering the \textbf{Maximum $k$-Vertex Protection} problem on degree bounded graphs. For the \textbf{Maximum $k$-Vertex Protection} problem, we call a strategy $S$ satisfying if the strategy $S$ contains at most $k$ vertices, and it saves maximum vertices on the graph $G$ from the fire. Note that the fire can no longer spread at time step t+1 after the firefighter protects the last vertex $v_t$, therefore \textbf{Maximum $k$-Vertex Protection} can be also regarded as a parameterized version of the firefighter problem when we bound the total number of time step $t \le k$. \\

First, we define a type of subgraph called $BFS\;Burning\;Tree$. Given the firefighter problem on a graph $G$ with a strategy $S$ with a sequence of protected vertices $V^*$, interconnect the source $s$ and vertex set $V^*$ by a tree $T$ of shortest $tree\;weight$, where the $tree\;weight$ is defined as the sum over the pairwise distances between the source $s$ and each vertex $v_i \in V^*$ on the subgraph $G-V^*\cup \{v_i\}$, with constraint that all vertices $V^*$ have to be a leaf of the tree $T$. Then $T$ is known as the $BFS\;Burning\;Tree$ of the strategy $S$. If the source $s$ and the vertex set $V^*$ are given, the $BFS\;Burning\;Tree$ of the strategy $S$ can be found in linear time by doing Breath-First Search starting from the source $s$, until all the vertices $V^*$ are reached.\\

A $Minimum\;BFS\;Burning\;Tree$ of a strategy $S$ on the graph $G$ is the $BFS\;Burning\;Tree$ with minimum number of vertices. Therefore, on a $Minimum\;BFS\;Burning\;Tree$ of a strategy $S'$, all the unnecessary edges and vertices which helps nothing on connecting the source $s$ with $S'$ have to be removed from the tree. From this definition, we can see that a $Minimum\;BFS\;Burning\;Tree$ of a strategy $S'$ is consist of either $|S'|+1$ or $|S'|$ leaves. Given that the number of protected vertices of an optimal strategy $|S_0|$ is at most $k$, we have the following lemma:

\begin{Lem}
In the firefighter problem with an optimal strategy $S_0$ of at most k vertices, the number of vertices of the $Minimum\;BFS\;Burning\;Tree$ of $S_0$ is at most $k^2+1$.
\end{Lem}

\begin{proof}
Since $S_0$ is an optimal strategy with at most k vertices, for each vertex $v_i \in S_0$, the distance from the source $s$ to $v_i$ on the subgraph $G-S_0\cup \{v_i\}$ is at most $k$. Otherwise the fire cannot reach $v_i$ before time step $k$+1, so the firefighter does not need to protect $v_i$ before time step k+1. Then the firefighter can earn one more turn to protect an extra vertex $v' \notin S_0$ on the $i$-th round, and protect $v_i$ on the $k+1$-th round, which contradicts with the assumption that $S_0$ is an optimal strategy.\\

Since a $Minimum\;BFS\;Burning\;Tree$ of $S_0$ is a tree minimizing the $tree\;weight$, which is the sum over all pairwise distances between the source $s$ and each vertex $v_i \in S_0$ on the subgraph $G-S_0\cup \{v_i\}$. As the distance from the source $s$ to $v_i$ on the subgraph $G-S_0\cup \{v_i\}$ is at most $k$, the $tree\;weight$ is at most $k|S_0|\le k^2$. And the number of edges of the $Minimum\;BFS\;Burning\;Tree$ is upper bounded by the $tree\;weight$, so the number of vertices of it is at most $k^2+1$.
\end{proof}

Using Lemma 1, we can design a randomized FPT algorithm for the \textbf{Maximum $k$-Vertex Protection} on degree bounded graph using random separation. The main idea of the algorithm is to randomly separate the graph G into 3 parts, one part contains the optimal strategy $S_0$, one part contains the source $s$ and all the internal nodes of the $Minimum\;BFS\;Burning\;Tree$ of $S_0$, the third part contains the neighbours of ($Minimum\;BFS\;Burning\;Tree$ of $S_0 - S_0$). \\

The algorithm first colors the source $s$ red, then colors each vertex of G randomly and independently by either green, red or yellow with probability $1/3$. We call a coloring of G a $good\;coloring$ if there exists a maximum k-vertex optimal strategy $S_0$ such that all vertices in $S_0$ are green, and the tree $T'= (Minimum\;BFS\;Burning\;Tree$ of $S_0 - S_0)$ are red and the neighbours of the tree $T'$ are yellow. Notice that the number of vertices in $S_0$ is at most k. By Lemma 1, we also know that the number of vertices in $T'$ is at most $k^2-k+1$. If the graph $G$ is a degree bounded graph with degree at most $d$, then $|N(T')|\le k^2d$. By random separation, we can have a $good\;coloring$ with probability at least $2^{-k^2(d+1)}$. \\

\begin{center}
\includegraphics{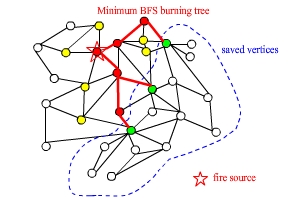}\\
\underline{Figure 2: A good coloring example on degree bounded graph}
\end{center}

Given a good coloring of $G$, we can find the specific satisfying strategy $S_0$ as the following. First, starting from the source $s$, do breath-first-search only on the red vertices to get a red subgraph $V'$ rooted at s. Since $T'$ is a red tree containing the source $s$, and all the neighbour of this red tree $T'$ are either yellow or green, but not red. Therefore, doing BFS from $s$ can help us to locate the tree $T'$, i.e. $V'=T'$. It takes linear time on number of vertices and edges of the subgraph $T'$. After the BFS, all of the green neighbours of $T'$ must be a satisfying strategy if the coloring is good, as all the non-solution neighbours of the red tree $N(T')-S_0$ are all colored yellow. Let $V^*$ be the green neighbours of T'. If $|V^*|\le k$, we check if the green vertex set $V^*$ is a valid strategy set for the problem using the {\bf  Procedure for verifying if $V^*$ is a satisfying strategy set} introduced in Section 3, with the input burnt vertex set $V'=T'$. If the coloring is good, this procedure will output $V^*$ in correct order as the optimal firefighting strategy $S_0$. \\

Similar to the time complexity analysis in Section 3, we can easily show the complexity of the above randomized algorithm is O($3^{k^2(d+1)}(n+m)$).


\section{Maximum $k$-Vertex Protection on unicyclic graphs}

In this Section, we are considering the \textbf{Maximum $k$-Vertex Protection} problem on $unicyclic$ $graphs$. A $unicyclic\;graph$ is a graph containing only 1 cycle. A $unicyclic\;graph$ can also be seen as a tree plus one extra edge. Given a unicyclic graph $G$ with a fire source $s$, let n be the number of vertices and edges of $G$. We would like to locate the optimal strategy $S_0$ with at most k vertices which saves maximum number of vertices.\\

In order to solve the problem on unicyclic graphs, we need to refer to a previous result of \textbf{Maximum $k$-Vertex Protection} problem on trees, given by Cai, Verbin and Yang\cite{CVY08}. In their paper, they introduce a FPT algorithm using $random\;separation$ to solve the \textbf{Maximum $k$-Vertex Protection} firefighter problem on trees. Their algorithm randomly and independently colors each vertex into either green or red. For a $good\;coloring$, all vertices in the satisfying strategy $S_0$ are green and all of their ascendants of vertices in $S_0$ on the tree are red. Given a $good\;coloring$, they can use greedy method to locate the k solution one-by-one in correct order by considering each level of the tree. They proved that this randomized algorithm runs in time $k^{-O(k)}n$. And it can be derandomized using $asymmetric\;universal\;sets$ introduced by Verbin \cite{Ver}, which helps us to get a deterministic algorithm that runs in time $k^{-O(k)}n\log n$.\\

We are going to use their result to show the \textbf{Maximum $k$-Vertex Protection} firefighter problem on unicyclic graphs are also FPT. The main idea of our algorithm is to separate the original problem into three cases, and handle them one by one. In each case, by observing the sequence of vertices being burnt, we find to transform the unicyclic graph $G$ to a tree $T'$ by removing vertices and edges on the cycle. Finally we can use the FPT algorithm given by Cai, Verbin and Yang to solve the firefighter problem on tree $T'$ and get the optimal solution of the graph $G$.\\

In the beginning, we need to find the cycle $C$ by doing Depth-First-Search from the source $s$. It takes only linear time. Then we can locate the vertex set $C$ of the cycle and the minimum distance from $s$ to cycle $C$. Let $l$ be the minimum distance from $s$ to cycle $C$. Generally, the fire source does not have to locate on the cycle, i.e. $l \ge 0$.\\

 Let $c_0$ be the vertex on the cycle $C$ which is nearest to the source $s$. Notice that $c_0$ is unique. If $l=0$, then the fire source $s$ locate on the cycle and $c_0=s$. Define path P to be the path connecting the source $s$ to $c_0$, where path P contains $l$ edges and $l+1$ vertices, including $s$ and $c_0$. Let $C$ contains r+1 vertices, where $2\le r \le n-1$. (If $r=0$ or $r=1$, then $G$ is not a simple graph. It contains either a loop or multiple edges.) We name the sequence of vertices on cycle $C$ to be $\{c_0,c_1,c_2,...,c_r\}$, starting with $c_0$ with clockwise direction.\\

Among all the vertices on the cycle $C$, $c_0$ is the one closest to the fire. If $c_0$ is saved, all the $c_i$ can be saved as $c_0$ block the only route that the fire can enter cycle $C$. It means that the fire can enter cycle C only if $c_0$ is burnt. Therefore, we try to use two FPT algorithms to consider the problem in the 3 cases below separately, and then combine them to get the optimal solution.


\begin{Case}
$c_0$ and the whole cycle $C$ will be saved at the end of the game.
\end{Case}

In case 1, we assume $c_0$ and the whole cycle $C$ will be saved at the end no matter what happens. To satisfy the criteria of this case, we modify graph G into a graph $T_1$, by replacing the whole cycle $C$ by a single vertex $c_0'$, linking to $2n$ individual vertices by $2n$ edges. Therefore, if G contains n vertices and n edges, $T_1$ is a tree containing $3n-r$ vertices and $3n-r-1$ edges. Then we have to use the following fact:

\begin{Fact}
The optimal solution $S_0$ of \textbf{Maximum $k$-Vertex Protection} firefighter problem on $unicyclic$ $graphs$ $G$, with constraint that $c_0$ has to be saved at the end of the game, is equivalent to the optimal solution $S_1$ of \textbf{Maximum $k$-Vertex Protection} firefighter problem on tree $T_1$.
\end{Fact}

Fact 2 is obviously true as $c_0'$ in $T_1$ is so important that a firefighter must keep it away from the fire, otherwises at least $2n$ vertices will be burnt. Therefore the 2 situations is totally equivalent. By solving the \textbf{Maximum $k$-Vertex Protection} on the tree $T_1$ using Cai, Verbin and Yang's FPT algorithm, we can get the solution $S_0$ of the problem on unicyclic graph $G$ based on the requirement of case 1.

\begin{center}
\includegraphics{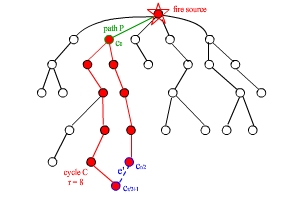}\\
\underline{Figure 3: Whole cycle C will be burnt}
\end{center}

\begin{Case}
$c_0$ and the whole cycle $C$ will be burnt at the end of the game.
\end{Case}

 In case 2, we assume the whole cycle $C$ will be burnt at the end, therefore the firefighter should not protect any vertices on path $P$ and cycle $C$. Since the fire will reach cycle $C$ starting from $c_0$, and none of the vertices on cycle $C$ get protected, so we have the following fact:

\begin{Fact}
If the fire burnt all the vertices on cycle $C$, the last vertex being burnt on cycle C, is either $c_{\lfloor r/2 \rfloor}$ or $c_{\lfloor r/2 \rfloor+1}$.
\end{Fact}

Let $e'$ be the edge on the cycle $C$ connecting the 2 vertices $c_{\lfloor r/2 \rfloor}$ and $c_{\lfloor r/2 \rfloor+1}$. When the fire pass along edge $e'$, the whole cycle $C$ is completely burnt already. Therefore, we can first transform G to a tree $T_2$ by deleting edge $e'$. Then we use random separation method to solve the \textbf{Maximum $k$-Vertex Protection} firefighter problem on the tree $T_2$. To guarantee the whole cycle $C$ will be burnt at the end, we colors all the vertices on the path $P$ and cycle $C$ red, before randomly coloring all the other vertices in green or red. Then the optimal solution $S_2$ of the problem on the tree $T_2$ is exactly the optimal solution $S_0$ of the problem on the graph $G$ based on the requirement of case 2.


\begin{Case}
$c_0$ will be burnt but at least one vertex on the cycle $C$ will be saved at the end of the game.
\end{Case}

This is the most complicated case to consider among all the cases. As $c_0$ will be burnt, which means the whole path P will be burnt, and the fire can enter cycle $C$ and burn some of the vertices on $C$. However, the criteria "at least one vertex on the cycle $C$ will be saved" tells us that there exists at least one protected vertex on the cycle $C$, unless no other protected vertices can save vertices on $C$ from the fire.\\

The main idea to solve this case is to exhaust all the possible combinations of protected vertices located on cycle $C$ in the optimal solution $S_0$, then we can figure out which vertices on cycle C will be saved, and delete them from the graph to get a tree. To ensure the algorithm runs in FPT time, we need the following important fact:

\begin{Fact}
Let $S_0$ be the optimal strategy of the problem on the unicyclic graph $G$ with a cycle $C$, then the number of vertices in $C\cap S_0$ is at most 2.
\end{Fact}

\begin{proof}
If $c_0$ is saved, we only need to protect one vertex on the path P in order to save the whole cycle $C$, therefore we never need to protect any vertices on $C$ except for $c_0$, then at most one protected vertex on $C$ would be enough.\\

Consider the case that $c_0$ is burnt. Let $q$ be the number of vertices in $S_0$ located on cycle $C$. We call those vertices $U = \{u_1,u_2,...,u_q\}\subseteq S_0$. Assume $q \ge 3$. We can always find two vertices in $U$, say $u_1$ and $u_2$, cut the cycle $C$ into two half. Let's call these two half $C_{up}$ and $C_{down}$, with $C_{up}$ contain the vertex $c_0$. \\

Consider another vertex $u_3$, if $u_3$ is on $C_{down}$, then $u_3$ is already saved by the 2 protected vertices $u_1$ and $u_2$ as they block the only two route that the fire can enter $C_{down}$. It is no point for the firefighter to include $u_3$ as a protected vertex in the optimal strategy. If $u_3$ is on $C_{up}$, then either $u_3$ and $u_1$ save $u_2$, or $u_3$ and $u_2$ save $u_1$ from the fire, either $u_1$ or $u_2$ becomes pointless and can be removed from the optimal strategy $S_0$. i.e. the set $U$ should contain at most 2 vertices.
\end{proof}

Fact 4 tells us that only 2 protected vertices on the cycle $C$ is sufficient for us to construct an optimal strategy $S_0$. Besides, by the requirement of case 3, at least one vertex on the cycle $C$ has to be protected. Exhaust for all pairs of vertices $u_1$ and $u_2$ on the cycle $C$, where $u_1$ can be equal to $u_2$. We assume $u_1$ and $u_2$ are the only protected vertices in $S_0$ on the cycle $C$, and they cut the cycle $C$ into two half $C_{up}$ and $C_{down}$, with $c_0$ located on the half $C_{up}$. We can easily see that the vertices on $C_{up}$ will be all burnt and the vertices on $C_{down}$ will be all saved. \\

For each pair of $u_1,u_2$,transform the graph $G$ to a tree $T_3(u_1,u_2)$ by removing all the vertices and edge on $C_{down}$, and also all the subtrees originally rooted at the vertices on $C_{down}$. Then we can use random separation to solve the problem on the graph $T_3(u_1,u_2)$, by pre-coloring the paths $P$ and $C_{up}$ to red, and $u_1,u_2$ to blue. Then the optimal strategy $S_3(u_1,u_2)$ of the problem on $T_3(u_1,u_2)$ is exactly the optimal strategy of the problem on graph $G$, condition on the criteria of case 3 and both $u_1$ and $u_2$ are the element of $S_0$.\\
 
\begin{center}
\includegraphics{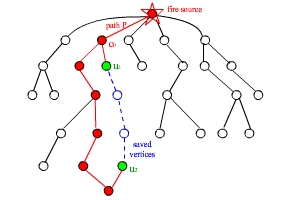}\\
\underline{Figure 4: Exhaustion for vertex $u_1$ and $u_2$}
\end{center}

In the exhaustion in case 3 , the choices of $u_1$ and $u_2$ are both bounded by the size of the cycle. Therefore we need to use Cai, Verbin and Yang's FPT algorithm O($n^2$) times in each choice of $u_1,u_2$. As all the transformations from a graph $G$ to a tree $T'$ mentioned above take only linear time. Combining the algorithms in all 3 cases, finally we construct a FPT algorithm runs in time $k^{-O(k)}n^3$. And it can be derandomized using $asymmetric\;universal\;sets$ to get a deterministic algorithm that runs in time $k^{-O(k)}n^3\log n$.\\

Actually, the algorithm above can be generalized to solve the firefighter problem in the family of graph $tree + b\;edges$, for any small constant $b$. In unicyclic graph, $b=1$, and the randomized FPT algorithm runs in time $k^{-O(k)}n^3$. Generally, one can use a similar approach to construct an algorithm solving the firefighter problem on $tree + b\;edges$ in time $k^{-O(k)}n^{(2b+1)}$. Unfortunately, it is not FPT time if both $k$ and $b$ are parameters of problem, as the complexity depends on $n^{(2b+1)}$, but the algorithm is also efficient when $b$ is a small fixed constant.\\

Here we omit the details on constructing algorithm of the \textbf{Maximum $k$-Vertex Protection} firefighter problem on $tree + b\;edges$ in $k^{-O(k)}n^{(2b+1)\log n}$ time. The main idea is to modify the exhaustion in case 3, from 2 vertices $u_1,u_2$ on cycle $C$ to $2b$ vertices $u_1,u_2,...,u_{2b}$ on the cyclic subgraph based on the fact below:

\begin{Fact}
Let $S_0$ be the optimal strategy of the problem on the $tree + b\;edges$ graph $G$, there exists a subgraph $G'$ where $G-G'$ is a forest, and the number of vertices in $G'\cap S_0$ is at most 2b.
\end{Fact}

The above fact can be deduced by fact 4 by considering the $b$ smallest cycles induced by the extra $b$ edges.\\

As Cai, Verbin and Yang show in their paper that the \textbf{Saving k Vertices} firefighter problem on trees has a polynomial kernel of size $k^2$ \cite{CVY08}, one may be curious if the problem on $tree + b\;edges$ graph also has a polynomial kernel or not. If there exists a polynomial kernel, we can reduce the complexity of the above algorithm to FPT time even for large b. However, since both the \textbf{Saving k Vertices} and \textbf{Maximum $k$-Vertex Protection} firefighter problem on general were proven to be W[1]-hard, and any general graph can be consider as a $tree + b\;edges$ graph if $b$ is not bounded, therefore it seems very uneasy to find a polynomial kernel in this situation.


\section{Weighted and unequal valued graph}
In the previous discussion, we only consider the firefighter problem on unweighted graph, i.e. all the edges have equal distance, therefore the fire spreads out uniformly to the unprotected neighbours of the burnt vertices at each time step. Besides, all the vertices are assumed to be equally valuable, and all the edges are set to have no values. In this section, we will focus on the firefighter problem on weighted graph and unequal valued graph.\\

We define the weighted firefighter problem as follow: A fire initially breaks out at a vertex r on a weighted graph G. In each time step $i$, a firefighter chooses to protect one vertex $v_i$, which is not yet burnt. Then for each edge $e'=(u,v)$ with an integral edge weight $w_{e'}$, where $u$ is already on fire at time $i-w_{e'}$ and $v$ is not yet on fire, the fire spreads out from $u$ to $v$ along the edge e and burn $v$ at time $i$. The objective of the problem is to choose a sequence of vertices $\{v_i\}$ to protect, in order to save maximum number of vertices from the fire. The original firefighter problem on unweighted graph is the weighted firefighter problem with all the edge weights equal to 1.\\

For any weighted graph $G$ of maximum weight $w_m$, we try to transform $G$ to an unweighted graph $G'$ by local replacement. For each edge $e=(u,v)$ with edge weight $w_e=0$, we contracts these edges and combine each (u,v) pair to be a single vertex. For each edge $e=(u,v)$ with edge weight $w_e>1$, we replace it by a path $P$ of $w_e$ edges by adding $w_e-1$ vertices between u and v. If $G'$ follows certain properties, for instance, $G'$ is degree bounded or unicyclic, then we can try to find the solution of the problem on graph $G'$ using random separation by pre-coloring all the newly-added vertices to yellow (or red in unicyclic case). Then we can apply the known FPT algorithm on unweighted graph to the weighted firefighter problem, where the time complexity would be almost the same (with a multiplication factor of $w_m$).\\

Now we consider the case that each vertices on the graph $G$ have unequal value, i.e. some vertices are more valuable that we have more motivation to save them from the fire. Also, we assumes the edges contain values too, therefore we are motivated to save valuable edges. A edge $e=(u,v)$ is said to be saved if both vertices u and v are saved, otherwises it is a burnt edge. This model is more realistic in the application of firefighter problem on diseases or computer-viruses control. Here we assign each vertex $v$ a "value" $z_v$, each edge $e$ a "value" $z_e$, and the objective of the $valued-firefighter$ problem is to choose a sequence of vertices $\{v_i\}$ to protect, in order to maximize the total value Z of saved vertices and saved edges. The original equally valued firefighter problem is the valued firefighter problem with all the vertex values equal to 1 and all the edge values equal to 0. \\

For any valued graph $G$ of maximum value $z_m$, we try to transform $G$ to an equally valued graph $G'$ by local replacement. For each edge $e=(u,v)$ with edge value $z_e \ge 1$, we replace it by a path $P$ of $z_e+1$ "value 0" edges by adding $z_e$ "value 1" vertices between u and v. For each vertex $v$ with value $z_v > 1$, we replace it by a path $P$ of $z_v$ "value 1" vertices by adding $z_v-1$ "value 0" edges. After all, we can also try to find the solution of the problem on graph $G'$ using random separation by pre-coloring all the newly-added vertices to yellow (or red in unicyclic case), which runs in the same time complexity of the originally problem (with a multiplication factor of $z_m$), similar to what we mention above in the weighted graph case.\\

Notice that the above algorithm does not work if there exists vertex with value 0. Actually, in the special case of firefighter problem on trees, the well-known \textbf{Saving k Leaves} problem is special case of firefighter problem on valued graph where internal vertices have value 0. And the parametric dual \textbf{Saving all But k Leaves} is NP-complete even for k=0, shown by Finbow et. al. \cite{FKMR07}, so there is no FPT algorithm for the problem \textbf{Saving all But k Leaves} unless P=NP. In Section 3 of this report, we show that \textbf{k Vertices Burnt} of equally valued firefighter problem on general graph is FPT. If there exists a reduction method to reduce any valued graph $G$ with 0-value vertices to an equally valued graph, then we show \textbf{Saving all But k Leaves} is FPT and prove P=NP. Therefore, it is not surprising that the above algorithm cannot reduce a valued graph with 0-value vertices to equal valued graph.


\section{Multiple fire sources, multiple protection and multiple burning}
In this final section, we are considering the firefighter problem with multiple fire sources and multiple firefighters. Multiple fire sources means that there are more than one sources (say g sources) on fire initially. Multiple protection means that the firefighter can protect more than one vertices, say p vertices, before the fire spreads out in each time step. Multiple protection means that the fire spreads out by several units length, say h units, after firefighter's action in each time step.\\


The firefighter problem with multiple fire sources on general graph is nothing but just an ordinary firefighter problem by combining all the g fire sources into one single source. Let $S=\{s_1,s_2,...,s_g\}$ be the g sources on the graph $G$,  we define $G'=merge(G;S)$. The result of \textbf{At Most k Vertices Burnt} problem on general graph $G'$ in Section 3 is the same in the result of the problem on graph $G$. For the \textbf{Maximum $k$-Vertex Protection} problem on degree bounded graph, we can also combine all the g fire sources into one single source. The only difference is the maximum degree of the source in G' is no longer $d$, but becomes $gd$. Therefore the complexity of the randomized algorithm is O($3^{(k^2+g)(d+1)}(n+m)$).\\

\begin{center}
\includegraphics[width=90mm]{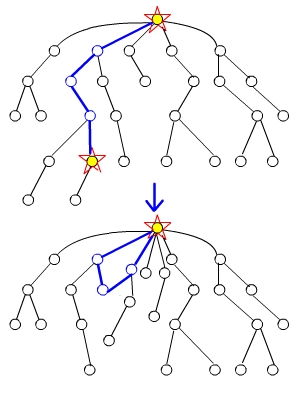}\\
\underline{Figure 5: Transformation from 2 fire sources on tree to 1 fire source on unicyclic graph}
\end{center}

Similarly, by combining all the g fire sources into one single source, the firefighter problem with multiple fire sources on $trees$ or $unicyclic\;graphs$ can be transform to the ordinary firefighter problem on $tree + (g-1)\;edges$ and $tree + g\;edges$ graph respectively. For example, consider the firefighter problem with 2 fire sources on a tree T,  by combining the two fire sources $s_1$ and $s_2$ into a single vertex, the tree T is transform to be an unicyclic graph $G'=merge(G;\{s_1,s_2\})$. Then we can use the FPT algorithm mentioned in Section 5 to solve it in $k^{-O(k)}n^3$ time. If $g$ is a small fixed constant, the firefighter problem with multiple fire sources on  $tree + b\;edges$ can be transformed to the firefighter problem with single fire source on $tree + (b+g-1)\;edges$, and there exists an algorithm runs in $k^{-O(k)}n^{(2b+2g-1)}$ time.\\


The firefighter problem with multiple protection and multiple burning is also very simple. By changing the parameter $k$ to be total number of rounds, we define a firefighter problem called  \textbf{Maximum $k$-Step Protection} with multiple protection and multiple burning as follow: \\

We want to find an optimal strategy $S_0$ = $\{\{v_{11},...,v_{1p}\},...,\{v_{k1},...,v_{kp}\}\}$, where $k$ is the total number of time steps of the game and $p$ is the number of vertices a firefighter can protect in each time step. This optimal strategy $S_0$ can save maximum number of the vertices from the fire, where the fire burns $h$ layers of neighbours of the vertices on fire in each time step.\\

Without multiple protection and multiple burning, the problem \textbf{Maximum $k$-Step Protection} can be reduced to  \textbf{Maximum $k$-Vertex Protection}. Otherwise, we can solve this problem in FPT time using the similar strategy as above. For example, considering the \textbf{Maximum $k$-Step Protection} firefighter problem with multiple protection and multiple burning on degree bounded graphs, we can treat the problem as an ordinary \textbf{Maximum $\lceil kp/h \rceil$-Vertex Protection} firefighter problem. Using the algorithm provided on Section 4, with $\lceil kp/h \rceil$ strategy vertices green, $\lceil kp/h \rceil^2$ BFS Burning Tree vertices red, and $\lceil kp/h \rceil^2d$ neighbouring vertices yellow as good coloring. Then we can do a BFS from the source $s$, and check if all the green leaves satisfy the criteria of optimal strategy or not. Notice that the {\bf  Procedure for verifying if $V^*$ is a satisfying strategy set} given in Section 3 also has to be modified a little bit as below:\\

\begin{center}
\fbox{
\begin{minipage}[l1pt]{5.50in}
{\bf  Modified Procedure for verifying if $V^*$ is a satisfying strategy set for multiple protection and burning} \\

{\bf Input}: graph G, source s, integer p and h, vertex set $V^*$ to be verified, and burnt vertex set $V'$\\
{\bf Output}: If true, output the sequence of satisfying strategy in correct order. If false, output "no".
\begin{enumerate}
	\item Do BFS on the subgraph $V'$ starting from s, find the distance $d[s,v]$ of the each vertices $v$ on $V^*$ from the source $s$ on the subgraph $G'=\{N[V'],E(V')\cup E(V',V^*)\}$.  
	\item Sort the vertices $V^*$ according to $d[s,v]$ in ascending order. Store the sorted sequence in an array $\{v_i\}$ 
	\item for i = 1 to $|V^*|$ do\\
                      	 If the distance from the source of the i-th vertex {\bf satisfies the relation $\lceil d[s,v_i]/h \rceil< \lceil i/p \rceil$}, then the solution is not valid, output "no" and halt.
            \item Output the sequence $\{v_i\}$ as the corresponding firefighting strategy in correct order.
    
\end{enumerate}
\end{minipage}
}
\end{center}

Therefore, we can solve the \textbf{Maximum $k$-Step Protection} firefighter problem with multiple protection and multiple burning in the degree bounded graph in FPT time. Similarly, by modifying the algorithm in Section 5, one can also construct an FPT algorithm for the \textbf{Maximum $k$-Step Protection} in the unicyclic graph.

\section{Acknowledgments}
I gratefully acknowledge the support of Computer Science and Engineering Department, the Chinese University of Hong Kong. I have to thank Professor Leizhen Cai, for the innoviative teaching and offering continuing advice and guidance. I am also thankful to my master supervisor, Professor Shengyu Zhang, who have been giving me much support and encouragement on my study and research.


\bibliography{firefighter}

\begin{thebibliography}{FKMR07}

\bibitem[CCC06]{CCC06}
L.~Cai, S.M. Chan, and S.O. Chan.
\newblock Random separation: A new method for solving fixed-cardinality
  optimization problems.
\newblock {\em Bodlaender, H.L., Langston, M.A. (eds.) IWPEC 2006. LNCS},
  4169:239--250, 2006.
\newblock Springer, Heidelberg.

\bibitem[CVY08]{CVY08}
L.~Cai, E.~Verbin, and L.~Yang.
\newblock Firefighting on trees: (1-1/e)-approximation, fixed parameter
  tractability and a subexponential algorithm.
\newblock {\em S.H. Hong, H. Hagamochi, and T. Fukunaga (Eds.) ISAAC}, LNCS
  5369:258--269, 2008.

\bibitem[CW09]{C09}
L.~Cai and W.~Wang.
\newblock The surviving rate of a graph for the firefighter problem.
\newblock {\em SIAM Journal on Discrete Mathematics}, 2009.

\bibitem[DH07]{DH07}
M.~Develin and S.G. Hartke.
\newblock Fire containment in grids of dimension three and higher.
\newblock {\em Discrete Applied Mathematics}, 155(17):2257--2268, 2007.

\bibitem[FKMR07]{FKMR07}
S.~Finbow, A.~King, G.~MacGillivray, and R.~Rizzi.
\newblock The firefighter problem for graphs of maximum degree three.
\newblock {\em Discrete Mathematics}, 307(16):2094--2105, 2007.

\bibitem[Fog03]{Fog03}
P.~Fogarty.
\newblock Catching the fire on grids.
\newblock {\em M.Sc. Thesis, Department of Mathematics, University of Vermont},
  2003.

\bibitem[Har95]{Har95}
B.~Hartnell.
\newblock Firefighter! an application of domination.
\newblock {\em 24th Manitoba Conference on Combinatorial Mathematics and
  Computing}, pages 4--27, 1995.
\newblock 24th Manitoba Conference on Combinatorial Mathematics and Computing,
  University of Minitoba, Winnipeg, Canada.

\bibitem[HL00]{HL9}
B.~Hartnell and Q.~Li.
\newblock Firefighting on trees: how bad is the greedy algorithm?
\newblock {\em Congr. Numer}, 145:187--192, 2000.

\bibitem[MW03]{MW03}
G.~MacGillivray and P.~Wang.
\newblock On the firefighter problem.
\newblock {\em J. Combin. Math. Combin. Comput.}, 47:83--96, 2003.

\bibitem[NSS95]{NSS95}
M.~Naor, L.J. Schulman, and A.~Srinivasan.
\newblock Splitters and near-optimal derandomization.
\newblock {\em IEEE Symposium on Foundations of Computer Science}, pages
  182--191, 1995.

\bibitem[Ver]{Ver}
E.~Verbin.
\newblock Asymmetric universal sets (in preparation).

\bibitem[WM02]{WM02}
P.~Wang and S.~Moeller.
\newblock Fire control on graphs.
\newblock {\em Journal of Combinatorial Mathematics and Combinatorial
  Computing}, 41:19--34, 2002.

\end{thebibliography}
\bibliographystyle{amsplain}

\end{document}